\documentclass[11pt]{article}
\usepackage{amsmath}
\usepackage{fullpage}
\usepackage{amsthm}
\usepackage{amssymb}
\usepackage{color}
\usepackage{array,multirow}
\usepackage{algorithm}
\usepackage{algorithmic}
\usepackage{tikz}

\usepackage{macros}
\usepackage{accents}

\newcommand{\opt}{\ensuremath{\text{\textsf{OPT}}}}
\newcommand{\cost}{\ensuremath{\text{\textsf{COST}}}}

\newcommand{\ofl}{\ensuremath{\text{\footnotesize\textsf{OFL}}}}
\newcommand{\dist}{d}

\newtheorem{theorem}{Theorem}
\newtheorem{lemma}{Lemma}
\newtheorem{definition}{Definition}

\newtheorem{proposition}{Proposition}

\newtheorem{example}{Example}

\newcommand{\np}{h \log\Delta}
\newcommand{\npv}{h}
\newcommand{\cram }{c_r}
\newcommand{\ram}{\ensuremath{\text{\footnotesize\textsf{KM-RAM}}}}
\newcommand{\ramcoreset}{\ensuremath{\text{\footnotesize\textsf{CS-RAM}}}}
\providecommand{\email}[1]{\href{mailto:#1}{\nolinkurl{#1}\xspace}}

\title{Improved Algorithms for Time Decay Streams}

\author{
Vladimir Braverman\thanks{Department of Computer Science, Johns Hopkins University, Baltimore, MD, USA.
Email: \email{vova@cs.jhu.edu}.
This material is based upon work supported in part by the National
Science Foundation under Grant No. 1447639,
by the Google Faculty Award and by DARPA grant N660001-1-2-4014. Its contents are
solely the responsibility of
the authors and do not represent the official view of DARPA or the Department of Defense.
}
\and
Harry Lang\thanks{MIT CSAIL, Cambridge, MA, USA.
Email: \email{harry1@mit.edu}.}
\and
Enayat Ullah \thanks{Department of Computer Science, Johns Hopkins University, Baltimore, MD, USA.
Email: \email{enayat@jhu.edu}.}
\and
Samson Zhou\thanks{School of Informatics, Computing, and Engineering, Indiana University, Bloomington, IN, USA. 
Email: \email{samsonzhou@gmail.com}.}
}

\date{\today}

\begin{document}

\maketitle

\begin{abstract}
In the time-decay model for data streams, elements of an underlying data set arrive sequentially with the recently arrived elements being more important. 
A common approach for handling large data sets is to maintain a \emph{coreset}, a succinct summary of the processed data that allows approximate recovery of a predetermined query. 
We provide a general framework that takes any offline-coreset and gives a time-decay coreset for polynomial time decay functions.

We also consider the exponential time decay model for $k$-median clustering, where we provide a constant factor approximation algorithm that utilizes the online facility location algorithm. 
Our algorithm stores $\bigO{k\log(h\Delta)+h}$ points where $h$ is the half-life of the decay function and $\Delta$ is the aspect ratio of the dataset. Our techniques extend to $k$-means clustering and $M$-estimators as well.
\end{abstract}

\section{Introduction}
The \emph{streaming model} of computation has become an increasingly popular model for processing massive datasets. 
In this model, the data is presented sequentially, and the objective is to answer some pre-defined query. 
The overwhelmingly large size of the dataset imposes a number of restrictions on any algorithm designed to answer the pre-defined query. 
For example, a streaming algorithm is permitted only a few passes, or in many cases, only a single pass over the data. 
Moreover, the algorithm should also use space sublinear in, or even logarithmic in, the size of the data. 
For more details on the background and applications of the streaming model, \cite{BabcockBDMW02, Muthukrishnan05, Aggarwal07} provide excellent surveys. 

Informally, a coreset for a given problem is a small \emph{summary} of the dataset such that the \emph{cost} of any candidate solution on the coreset is approximately the same as the cost in the original set.
Coresets have been used in a variety of problems, including generalized facility locations \cite{FeldmanFS06}, $k$-means clustering \cite{FeldmanMS07, braverman2016new}, principal component analysis \cite{FeldmanSS13}, and $\ell_p$-regression \cite{DasguptaDHKM09}. 
Coresets also have a number of applications in distributed models (see \cite{IndykMMM14, MirrokniZ15, BarbosaENW16, AssadiK17}, for example). 
To maintain the coresets throughout the data stream, one possible approach is the so called merge-and-reduce method, in which the multiple sets may be adjusted and combined. 
Several well-known coreset constructions \cite{Har-PeledM04, chen2009coresets} for the $k$-median and $k$-means problems are based on the merge-and-reduce paradigm.

\subsection{Motivation}
Many applications discard obsolete data, choosing to favor relatively recent data to base their queries. 
This motivates the \emph{time decay} model, in which there exists a function $w$ so that the weight of the $t^\text{th}$ most recent item is $w(t)$. 
Note that this is a generalization of both the \emph{insertion-only} streaming model, where $w(t)=1$ for all $t$, and the \emph{sliding-window} model, where $w(t)=1$ for the most recent $W$ items, and $w(t)=0$ for $t>W$. 
In this paper, we study the problem of maintaining coresets over a polynomial decay model, where $w(t)=\frac{1}{t^s}$ for some parameter $s>0$, and an exponential decay model, where $w(t)=2^{\frac{T-t+1}{h}}$ at time $T$ for some \emph{half-life} parameter $h>0$. 

Although exponential decay model is well-motivated by natural phenomena that exhibit half-life behavior, \cite{CohenS03} notices that exponential decay and the sliding window model is often insufficient for many applications because the decay occurs too quickly and suggests that polynomial decay may be a reasonable alternative for some applications, such as availability of network links. 
For example, consider a network link that fails at every time between $10$ and $60$ and a second network link that fails once at time $75$. 
Intuitively, it seems like the second link should be better, but under many parameters, the exponential decay model and sliding window model will both agree that the first link is better. 
Fortunately, under the polynomial decay model, events that occur near the same time have approximately the same weight, and we will obtain some view in which the first link is preferred \cite{KopelowitzP05}. 
In practice, time decay functions have been used in natural language understanding to give more importance to recent utterances than the past ones \cite{su2018time}.

\paragraph{Organization.}The rest of the paper is organized as follows. In Section \ref{sec:results}, we summarize the main results of the paper and the algorithmic approaches. In Section \ref{sec:related}, we discuss the related work, and in Section \ref{sec:prelim}, we formalize the problem and discuss the preliminaries required. In Sections \ref{sec:poly} and \ref{sec:exp}, we handle the polynomial and exponential decay, respectively, in detail, wherein we present  the algorithmic details as well as the complete analysis.

\section{Our Contributions}
\label{sec:results}
We summarize our results and give a high-level idea of our approach for problems in the polynomial and exponential decay models in the following subsections respectively. 
The reader is encouraged to go through Sections \ref{sec:poly} and \ref{sec:exp} for details.

\subsection{Polynomial decay}
In the polynomial decay model, a stream of points $P$ arrives sequentially and the weight of the $t^\text{th}$ most recent point, denoted as $w(t)$, is $w(t) = \frac{1}{t^s}$ where $s>0$ is a given constant parameter of the decay function. 
We first state a theorem that shows that we can use an offline coreset construction mechanism to give a coreset for the polynomial decay model.

\begin{theorem}
\label{thm:poly} Given an algorithm that takes a set of $n$ points as input and constructs an $\epsilon$-coreset of $F(n,\epsilon)$ points in $\bigO{n T(\epsilon)}$ time, there exists a polynomial decay algorithm that maintains an $\epsilon$-coreset while storing $\bigO{\epsilon^{-1} \log n \, F\left(n, \frac{\epsilon}{\log n}\right)}$ points and with $\bigO{\epsilon^{-1}\ \log n \ F(n,\epsilon) \, T(\epsilon/\log n)}$ update time.
\end{theorem}
Theorem \ref{thm:poly} applies to any time-decay problem on data streams that admits an approximation algorithm using coresets. 
Among its applications are the problems of $k$-median and $k$-means clustering, $M$-estimator clustering, projective clustering, and subspace approximation. 
We list a few of these results in Table \ref{table:coresets}.
Our result is a generalization of the \emph{vanilla} merge-and-reduce approach used to convert offline coresets to streaming counterparts. 
In particular, plugging in $s=0$, we get the vanilla streaming model, and the theorem recovers the corresponding guarantees.

\begin{table}[H]
    \centering
    \begin{tabular}{|c|c|c|}
    \hline
    Problem & Coreset size & Offline algorithm \\
    \hline 
    \hline
    Metric $k$-median clustering & $\bigO{\frac{s}{\epsilon^3}k\log k\log^4 n}$ & \cite{FeldmanL11} \\
    \hline
    Metric $k$-means clustering & $\bigO{\frac{s}{\epsilon^3}k\log k\log^4 n}$ & \cite{braverman2016new}\\
     \hline
    Metric $M$-estimator & $\bigO{\frac{s}{\epsilon^3}k\log k\log^4 n}$ & \cite{braverman2016new}\\
    \hline
   $j^\text{th}$ subspace approximation  & $\bigO{\frac{j^2s}{\epsilon^4}\log^8 n\log\br{\frac{\log n}{\epsilon}}}$ & \cite{FeldmanL11} \\
    \hline
    Low rank approximation  & $\bigO{\frac{s}{\epsilon^2} \, kd \log n}$ & \cite{ghashami2016frequent} \\
    \hline
    \end{tabular}
    \caption{Coresets for some problems in polynomial decay streams }
    \label{table:coresets}
\end{table}

\paragraph{Approach.}
A natural starting point would be to attempt to generalize existing sliding window algorithms to time decay models. 
These algorithms typically use a histogram data structure~\cite{BravermanO07}, in which multiple instances of streaming algorithms are started at various points in time, one of which well-approximates the objective evaluated on the data set represented by the sliding window. 
However, generalizing these histogram data structures to time-decay models does not seem to work since the weights of all data points changes upon each new update in time-decay model, whereas streaming algorithms typically assume static weights for each data point. 

Instead, our algorithm partitions the stream into blocks, where each block represents a disjoint collection of data point between certain time points. 
Each arriving element initially begins as its own block, containing one element. 
The algorithm maintains an unweighted coreset for each block, and merges blocks (i.e corresponding coresets) as they become older. 
However, at the end, each block is to be weighted according to some function, and so the algorithm chooses to merge blocks when the weights of the blocks become ``close''. 
Thus, a coreset for each block will represent the set of points well, as the weights of the points in each block do not differ by too much.


\subsection{Exponential decay}
We also provide an algorithm that achieves a constant approximation for $k$-median clustering in the exponential decay model. 
Our guarantees also extend to $k$-means clustering and $M$-estimators.

Given a set $P$ of points in a metric space, let $\Delta$ denote its aspect ratio i.e the ratio between the largest and (non-zero) smallest distance between any two points in $P$. 
The weight of the $t^\text{th}$ most recent point at time $T$ is $w(t) = 2^{\frac{T-t+1}{h}}$ where $h > 0$ is the half-life parameter of the exponential decay function.
\begin{theorem}
 \label{thm:exp}
There exists a streaming algorithm that given a stream $P$ of points with exponentially decaying weights, with aspect ratio $\Delta$ and half-life $h$, produces an $\bigO{1}$-approximate solution to k-median clustering. The algorithm runs in $\bigO{nk\log (h \Delta)}$ time and uses $\bigO{k \log (h\Delta) + h}$ space.
\end{theorem}

\paragraph{Approach.}
Although our previous framework will work for other decay models, the algorithm may use prohibitively large space. 
The intuition behind the polynomial decay approach is that a separate coreset is maintained for each set of points that roughly have the same weight. 
In other words, the previous framework maintains a separate coreset each time the weight of the points decrease by some constant amount, so that if $R$ is the ratio between the largest weight and the smallest weight, then the total number of coresets stored by the algorithm is roughly $\log R$. 
In the polynomial decay model, the number of stored coresets is $\bigO{\log n}$, but in the exponential decay model, the number of stored coresets would be $\bigO{n}$, which would no longer be sublinear in the size of the input. 
Hence, we require a new approach for the exponential decay model.

Instead, we use the online facility location (\ofl) \ algorithm of Meyerson \cite{meyerson2001online} as a subroutine to solve $k$-median clustering in the exponential decay model. 
In the online facility location problem, we are given a metric space along with a facility cost for each point/location that appears in the data stream. 
The objective is to choose a (small) number of facility locations to minimize the total facility cost plus the service cost, where the service cost of a point is its distance to the closest facility. For more details, please see Section \ref{sec:exp}.

Our algorithm for the exponential time decay model proceeds on the data stream, working in phases. 
Each phase corresponds to an increasing ``guess'' for the value of the \emph{cost} of the optimal clustering. 
Using this guess, each phase queries the corresponding instance of \ofl. 
If the guess is \emph{correct}, then the subroutine selects a bounded number of facilities. 
On the other hand, if either the cost or the number of selected facilities surpasses a certain quantity, then the guess for the optimal cost must be incorrect, and the algorithm triggers a phase change.
Upon a phase change, our algorithm uses an offline $k$-median clustering algorithm to cluster the facility set and produces exactly $k$ points.  It then runs a new instance of \ofl \ with a larger guess, and continues processing the data stream.

However, there is a slight subtlety in this analysis. 
The number of points stored by \ofl \ is dependent on the weights of the point. 
In an exponential decay function, the ratio between the largest weight and smallest weight of points in the data set may be exponentially large. 
Thus to avoid \ofl \ from keeping more than a logarithmic number of points, we force \ofl \ to terminate after seeing $\log(h\Delta)$ points during a phase. 
Furthermore, we store points verbatim until we see $k+h$ \emph{distinct} points, upon whence we will trigger a phase change. 
We show that forcing this phase change does indeed correspond with an increase in the guess of the value for the optimal cost.

\section{Related Work} \label{sec:related}
The first insertion-only streaming algorithm for the $k$-median clustering problem was presented in 2000 by Guha, Mishra, Motwani, and O'Callaghan \cite{guha2000clustering}. 
Their algorithm uses $\bigO{n^\epsilon}$ space for a $2^{\bigO{1/\epsilon}}$ approximation, for some $ 0<\epsilon<1$. 
Subsequently, Charikar \etal\, \cite{charikar2003better} present an $\bigO{1}$-approximation algorithm for $k$-means clustering that uses $\bigO{k\log^2 n}$ space. 
Their algorithm uses a number of phases, each corresponding to a different guess for the value of the cost of optimal solution. 
The guesses are then used in the online facility location (\ofl) algorithm of \cite{meyerson2001online}, which provides a set of centers whose number and cost allows the algorithm to reject or accept the guess. 
This technique is now one of the standard approaches for handling $k$-service problems. 
Braverman \etal\,\cite{braverman2011streaming} improve the space usage of this technique to $\bigO{k\log n}$.
\cite{BravermanLLM15} and \cite{BravermanLLM16} develop algorithms for $k$-means clustering on sliding windows, in which expired data should not be included in determining the cost of a solution.

Another line of approach for $k$-service problems is the construction of coresets, in particular when the data points lie in the Euclidean space. 
Har-Peled and Mazumdar \cite{Har-PeledM04} give an insertion-only streaming algorithm for $k$-medians and $k$-means that provides a $(1+\epsilon)$-approximation, using space $\bigO{k\epsilon^{-d}\log^{2d+2}n}$, where $d$ is the dimension of the space. 
Similarly, Chen \cite{chen2009coresets} introduced an algorithm using $\bigO{k^2d\epsilon^{-2}\log^8 n}$ space, with the same approximation guarantees. 

Cohen and Strauss \cite{CohenS03} study problems in time-decaying data streams in 2003. 
There are a number of results \cite{KopelowitzP05, CormodeTX07, CormodeKT08, CormodeTX09} in this line of work, but the most prominent time-decay model is the sliding window model. 
Datar \etal\,\cite{DatarGIM02} introduced the exponential histogram as a framework in the sliding window for estimating statistics such as count, sum of positive integers, average, and $\ell_p$ norms. 
This initiated an active line of research, including improvements to count and sum \cite{GibbonsT02}, frequent itemsets \cite{ChiWYM06, BravermanGLWZ18}, frequency counts and quantiles \cite{ArasuM04, LeeT06}, rarity and similarity \cite{DatarM02}, variance and $k$-medians \cite{BabcockDMO03} and other geometric and numerical linear algebra problems \cite{FeigenbaumKZ04, ChanS06,braverman2018numerical}.

\section{Preliminaries} \label{sec:prelim}
Let $\mathcal{X}$ be the set of possible points in a space with metric $d$.
A weighted set is a pair $(P,w)$ with a set $P \subset \mathcal{X}$ and a weight function $w : P \rightarrow [0,\infty)$.
A query space is a tuple $(P,w,f,Q)$ that combines a weighted set with a set $Q$ of possible queries and a function $f : \mathcal{X} \times Q \rightarrow [0,\infty)$.
A query space induces a function 
\[
\bar{f}(P,w,q) = \sum_{p \in P} w(p) f(p,q).
\]
We now instantiate the above with some simple examples.
\begin{example}[$k$-means]
Let $Q$ be all sets of $k$ points in $\mathbb{R}^d$, and for $C \in Q$ define $f(p,C) = \min_{c \in C} d^2(p,c)$. 
The $k$-means cost of $(P,w)$ to $C$ is
\[
\sum_{p \in P} w(p) \min_{c \in C} d^2(p,c).
\]
\end{example}

\begin{example}[$k$-median]
Let $Q$ be all sets of $k$ points in $\mathbb{R}^d$, and for $C \in Q$ define $f(p,C) = \min_{c \in C} d(p,c)$. 
The $k$-median cost of $(P,w)$ to $C$ is
\[
\sum_{p \in P} w(p) \min_{c \in C} d(p,c).
\]
\end{example}

\noindent
Note that both $k$-median and $k$-means are captured in $\bar{f}(P,w,C)$. We now define an $\epsilon$-coreset.

\begin{definition}[$\epsilon$-coreset] A $\epsilon$-coreset for the query space $(P,w,f,Q)$ is a tuple $(Z,u)$, where $Z \subseteq \mathcal{X}$ is a set of points and $u:Z\rightarrow[0,\infty)$ are their corresponding weights, such that for every $q \in Q$
\[(1-\epsilon) \bar f(P,w,q) \leq \bar f(Z,u,q) \leq (1+\epsilon) \bar f(P,w,q).\]
\end{definition}

An important property of coresets is that they are \emph{closed} under  operations like union and composition. We formalize this below.
\begin{proposition}[Merge-and-reduce]
\label{prop:coresets} \cite{chen2009coresets}
Coresets satisfy the following two properties.
\begin{enumerate}
\item If $\cS_1$ and $\cS_2$ are $\epsilon$-coresets of disjoint sets $\cP_1$ and $\cP_2$ respectively, then $\cS_1 \cup \cS_2$ is an $\epsilon$-coreset of $\cP_1 \cup \cP_2$.
\item If $\cS_1$ is an $\epsilon$-coreset of $\cS_2$ and $\cS_2$ is a $\delta$-coreset of $\cS_3$, then $\cS_1$ is a $((1+\epsilon)(1+\delta) - 1)$-coreset of $\cS_3$.
\end{enumerate} 
\end{proposition}
We now define approximate triangle inequality, a property that allows us to extend our results obtained in metric spaces to ones with \emph{semi-distance} functions. 
In particular, this allows us to extend results for $k$-median clustering to $k$-means and $M$-estimators in exponential decay streams.
\begin{definition}
[$\lambda$-approximate triangle inequality]
\label{def:ati}
A  function $d(\cdot,\cdot)$ on a space $\cX$ satisfies the $\lambda$-approximate triangle inequality if for all $x,y,z\in \cX$,
\[d(x,z)\le\lambda(d(x,y)+d(y,z)).\]
\end{definition}

\section{Polynomial decay} \label{sec:poly}
We consider a time decay, wherein a point $p$ in the stream, which arrived at time $t$, has weight $w(p) = (T-t+1)^{-s}$ at time $T>t$, for some parameter $s > 0$. Equivalently, the $t^{\text{th}}$ most recent element has weight $t^{-s}$ for some $s > 0$.  

We present a general framework which, for given problem, takes an offline coreset construction algorithm and adapts it to polynomial decay streams. Our technique can be viewed as a generalization of merge-and-reduce technique of Bentley and Saxe \cite{bentley1980decomposable}. We also briefly discuss some applications towards that end. We start with stating our main theorem for polynomial decay streams.
\begin{theorem}
\label{thm:poly2}
Given an offline algorithm that takes a  set of $n$ points as input and constructs an $\epsilon$-coreset of $F(n,\epsilon)$ points in $\bigO{n \, T(\epsilon)}$ time, there exists a polynomial decay algorithm that maintains an $\epsilon$-coreset while storing $\bigO{\epsilon^{-1}s\log n \, F(n,\epsilon/\log n)}$ points and with update time 
\[\bigO{\epsilon^{-1}s\ \log n \ F(n,\epsilon) \, T(\epsilon/\log n)}.\]
\end{theorem}

\paragraph{Notation.} We use $\bN$ to denote the set of natural numbers. We use \ramcoreset \ to denote an offline coreset construction algorithm, which given $n$ points, constructs an $\epsilon$-coreset in time $\bigO{n \, T(\epsilon)}$ and takes space $F(n,\epsilon)$. We abuse notation by using $F(n,\epsilon)$ to  also refer to the corresponding coreset.
\subsection{Algorithm} 
We start with giving a high-level intuition of the algorithm. Given a stream of points, the algorithm implicitly maintains a partition of the streams into disjoint \emph{blocks}. A block is a collection of consecutive points in the stream, and is represented by two positive integers $a$ and $b$ as $[a,b]$, where $a$ represents the position of the first point in the block and $b$ the last point, relative to the start of the stream. Let the set of blocks be denoted by $\cB$. Our algorithm stores points of a given block by maintaining a coreset for the points in that block. As the stream progresses, we merge older blocks i.e. the corresponding coresets. Informally, the merge happens when the weights of the blocks become \emph{close}.

We first define a set of integer \emph{markers} $x_i$, which for a given $i \in \bN$, depends on the decay parameter $s$ and target $\epsilon$. These markers dictate when to merge blocks as the stream progresses. For a given $i\in \bN$, we define $x_i$ to be the minimum integer greater than or equal to $2^i$ such that
\begin{align*}
    \frac{1-\epsilon}{(x_i - 2^i +1)^s} \leq \frac{1+\epsilon}{x_i^s}.
\end{align*}
Equivalently, we can write $\br{\frac{x_i}{x_i-2^i+1}}^s \leq \frac{1+\epsilon}{1-\epsilon}$. Note that each of the $2^i$ points following $x_i$ in the stream, has weight within $\frac{1+\epsilon}{1-\epsilon}$ times the weight of $x_i$. Moreover, $x_i$'s can be exactly pre-computed from the equation and we therefore assume that these are implicitly stored by the algorithm. Each new element in the stream starts as a new block. As mentioned before, the blocks are represented by two integers $[a,b]$ and the points are stored as a coreset. When a block $[a,b]$ \emph{reaches} $x_i$, then algorithm merges all of $[x_i-2^i+1,x_i]$ points into a single coreset. In the end, the algorithm outputs the \emph{weighted} union of the coresets of the blocks.

To visualize this, consider the integer line, and suppose that we have $x_i$'s marked on the positive side of the line, for example $x_1 = 2, x_2 = 4 \ldots$. The tuple indices of the blocks represent the relative position of the point in the stream, with the start being $1$ and the end point being $n$. At the start,  the stream is on the non-positive end with the first point at $0$. As the time progresses, the stream moves to the right side. Therefore, when we observe the first element, it moves to the point $1$. We then store it as a new block, represented by $[1,1]$; we also simultaneously store a coreset corresponding to it. As time progresses, a block reaches $x_i$ for some $i$ which can be formally expressed as $a+x_i \leq n$. We then merge all blocks in the range $[a,a+2^i-1]$. Note that by definition of $x_i$, we would have observed all these elements and also we will not merge partial blocks.  We present this idea in full in Algorithm \ref{alg:polydecay} and intuition in Figure~\ref{fig:merge}. We remark that when we construct coresets, we use an offline algorithm \ramcoreset \ which given a set of $n$ points $P$ and a query space $(P,w,f,q)$ produces an $\epsilon$-coreset. 

 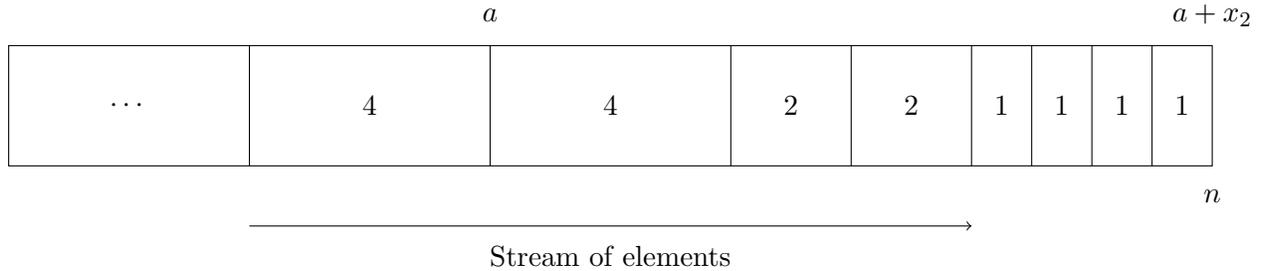
\begin{figure*}[!htb]
 \centering
 \begin{tikzpicture}[scale=0.8]
\draw (0cm,0cm) rectangle+(20cm,2cm);
 \node at (20cm,-0.5cm){$n$};
 \node at (20cm,2.5cm){$a+x_2$};
 \draw (19cm,0cm) -- (19cm, 2cm);
 \node at (19.5cm,1cm){$1$};
 \draw (18cm,0cm) -- (18cm, 2cm);
 \node at (18.5cm,1cm){$1$};
 \draw (17cm,0cm) -- (17cm, 2cm);
 \node at (17.5cm,1cm){$1$};
 \draw (16cm,0cm) -- (16cm, 2cm);
 \node at (16.5cm,1cm){$1$};
 \draw (14cm,0cm) -- (14cm, 2cm);
 \node at (15cm,1cm){$2$};

 \draw (12cm,0cm) -- (12cm, 2cm);
 \node at (13cm,1cm){$2$};
 \node at (8cm,2.5cm){$a$};
 \draw (8cm,0cm) -- (8cm, 2cm);
 \node at (10cm,1cm){$4$};
 \draw (4cm,0cm) -- (4cm, 2cm);
 \node at (6cm,1cm){$4$};
 \node at (2cm,1cm){$\cdots$};
 \draw[->] (4cm,-1cm) -- (16cm, -1cm);
 \node at (10cm, -1.5cm){Stream of elements};
 \end{tikzpicture}
 \caption{The algorithm merges blocks in each interval $[a,a+2^i-i]$ for $a\le n-x_i$}
 \label{fig:merge}
 \end{figure*}

\begin{algorithm}[!htb]
\caption{$\epsilon$-coreset for polynomial decaying streams}
\label{alg:polydecay}
\begin{algorithmic}[1]
\REQUIRE Stream $P$, polynomial decay function $w(t) = \frac{1}{t^s}$, for some $s>0$, an offline coreset construction algorithm \ramcoreset
\ENSURE $(1+\epsilon)$ coreset.
\STATE Initialize $\cB = \emptyset$ 
\FOR{each element $p_n$ of the stream}
    \STATE Insert $[n,n]$ into $\cB$ as a new block and construct a coreset
    \FOR{each block $[a,b] \in \cB$}
        \IF{$a + x_i < n$ for some $i$}
            \STATE \emph{Merge} the blocks in $[a,a+2^i-1]$ and \emph{reduce} to get an $\frac{\epsilon}{3\log n}$-coreset
        \ENDIF
    \ENDFOR
\ENDFOR
\FOR{each block $[a,b] \in \cB$}
    \STATE Give the block weight $u(a,b) = \frac{1}{2}\br{\frac{1-\epsilon}{a^s} + \frac{1+\epsilon}{b^s}}$
\ENDFOR
\end{algorithmic}
\end{algorithm}

\subsection{Analysis}
We first show that a weighted combination of blocks gives us an $\epsilon$-coreset. For a block $[a,b]$, let the weight of the block be denoted as $u(a,b)$. We set $u(a,b) = \bar{u}$ where $\bar{u}$ satisfies
\begin{align*}
    \frac{1-\epsilon}{a^s} \leq \bar{u} \leq \frac{1+\epsilon}{b^s}.
\end{align*}
The following lemma shows that any such $\bar{u}$ produces a $3\epsilon$-coreset.

\begin{lemma}
\label{lem:polydecay1}
Let $(Z,u)$ be an $\epsilon$-coreset for $(P,w,f,Q)$. Let $\bar{u}:Z\rightarrow [0,\infty)$ be such that $(1-\epsilon)u(z) \leq \bar{u}(z) \leq (1+\epsilon)u(z)$ for every $z \in Z$, then $(Z,\bar{u})$ is a $3\epsilon$-coreset for $(P,w,f,Q)$.
\begin{proof}
Since $(Z,u)$ is an $\epsilon$-coreset for $(P,w,f,Q)$, therefore for every $q \in Q$,
\begin{align*}
    & (1-\epsilon)\bar{f}(P,w,q) \leq \bar{f}(Z,u,q) \leq (1+\epsilon)\bar{f}(P,w,q) \\
    \iff & (1-\epsilon)\sum_{p \in P}w(p)f(p,q) \leq \sum_{z \in Z}u(z)f(z,q) \leq (1+\epsilon)\sum_{p \in P}w(p)f(p,q) \\
    \iff & (1-\epsilon)^2 \sum_{p \in P}w(p)f(p,q) \leq \sum_{z \in Z}\bar{u}(z)f(z,q) \leq (1+\epsilon)^2\sum_{p \in P}w(p)f(p,q).
\end{align*}
Note that for $\epsilon < 1$, we have $(1-2\epsilon)\bar{f}(P,w,q) \leq (1-\epsilon)^2\bar{f}(P,w,q) \leq \bar{f}(Z,\bar{u},q) \leq (1+\epsilon)^2 \bar{f}(P,w,q)\leq (1+3\epsilon)f(P,w,q)$. Therefore $(Z, \bar{u})$ is a $3 \epsilon$-coreset for $(P,w,f,Q)$.
\end{proof}
\end{lemma}

Having assigned weights to the blocks, we can take the union to get the coreset of $\cB$. For simplicity, we choose $u(a,b) = \frac{1}{2}\br{\frac{1-\epsilon}{a^s} + \frac{1+\epsilon}{b^s}}$ in Algorithm \ref{alg:polydecay}. 
We now present a lemma that bounds the number of blocks maintained by the algorithm.

\begin{lemma}
\label{lem:polydecayspace}
Given a polynomial decay stream of $n$ points as input to Algorithm \ref{alg:polydecay}, the number of blocks produced is $\bigO{\epsilon^{-1}s\log n}$.
\end{lemma}
\begin{proof}
Consider any two adjacent blocks. 
By the definition of the $x_i$'s, the ratio between the weights of the oldest and youngest elements is at least $(1+\epsilon)/(1-\epsilon)$. 
In the full stream, the oldest element has weight $1/n^s$ and the youngest element has weight $1$. 
Let $B$ be the number of blocks so that $\br{\frac{1+\epsilon}{1-\epsilon}}^{\floor{B}} \leq n^s$. 
Solving for $B$, we get $B \leq  \frac{ s \log n}{\log\br{(1+\epsilon)/(1-\epsilon)}}$. 
We will now lower bound the denominator using the numerical inequality $\text{ln}(1+x) \geq \frac{2x}{2+x}$ for $x>0$; equivalently $\log(1+x) \geq c \cdot \frac{2x}{2+x}$ for $x>0$ and $c=\Theta(1)$. We get $\log\br{\frac{1+\epsilon}{1-\epsilon}} = \log\br{1 + \frac{2\epsilon}{1-\epsilon}} \geq 2c\epsilon$, and therefore we have $B = \bigO{\epsilon^{-1}s\log n}$.
\end{proof}

We now give the proof of the main theorem for the polynomial decay model.

\begin{proof}[Proof of Theorem \ref{thm:poly2}]
From Proposition \ref{prop:coresets}, we get that when we merge disjoint blocks, we do not sacrifice the coreset approximation parameter $\epsilon$. However, when we reduce, for instance two $\epsilon$-corsets, we get a $2\epsilon$-coreset. For $n$ points observed in the stream, note that there would be at most $\log n$ reduces. This follows from the fact that the size of successive blocks increase exponentially. Therefore using an offline $\epsilon'$-coreset construction algorithm \ramcoreset \ with $\epsilon' = \epsilon/3\log n$, we get that merging and reducing the blocks produces an $\epsilon/3$-coreset (by Proposition \ref{prop:coresets}). Finally, from Lemma \ref{lem:polydecay1}, we get that taking a union of these blocks weighted by $u(a,b) = \frac{1}{2}\br{\frac{1-\epsilon}{a^s} + \frac{1+\epsilon}{b^s}}$ gives us an $\epsilon$-coreset.

For the space bound, we have from Lemma \ref{lem:polydecayspace} that the number of blocks is $\bigO{\epsilon^{-1}s\log n}$. Since we maintain an $\epsilon/\log n$ coreset for each block, we get that the offline coreset construction algorithm takes space $F(n,\epsilon/\log n)$. Therefore, we get that the space complexity is $\bigO{\epsilon^{-1}s\log n \, F(n,\epsilon/\log n)}$. For update time, note that for $n$ points, we have $\bigO{\epsilon^{-1}s\log n}$ blocks and we use an $\br{\epsilon/\log n}$-coreset algorithm which takes time $\bigO{F(n,\epsilon) \ T(\epsilon/\log n)}$ per block. We therefore get a total time of $\bigO{\epsilon^{-1}s\ \log n \ F(n,\epsilon) \, T(\epsilon/\log n)}$
\end{proof}

\paragraph{Applications.} Coresets have been designed for a wide variety of geometric, numerical linear algebra and learning problems. Some examples include $k$-median and $k$-means clustering \cite{chen2009coresets}, low rank approximation \cite{sarlos2006improved}, $\ell_p$ regression \cite{clarkson2009numerical},  projective clustering \cite{deshpande2006matrix}, subspace approximation \cite{feldman2010coresets}, kernel methods \cite{zheng2017coresets}, Bayesian inference \cite{huggins2016coresets} etc. We instantiate our framework with a few of these problems, and present the results in Table \ref{table:coresets}.

\section{Exponential decay} \label{sec:exp}
We now discuss another model of time decay in which the weights of previous points decay exponentially with time. 
Analogous to our polynomial decay model, a point that first appeared in the stream at time $t \leq T $ has weight $2^{\frac{T-t+1}{h}}$ at time $T$, where the parameter $h > 0$ is the half-life of the decay function. 
We however consider a different viewpoint to simplify the analysis; we maintain that the weight of a point observed at time $t$ is \emph{fixed} to be $2^{t/h}$ where $h>0$ is the half-life parameter. 
These are equivalent since the ratio of weights between successive points is the same in both the models. 

\paragraph{Online Facility Location.} We first discuss the problem of Online Facility Location (\ofl) as our algorithm uses it as a sub-routine. The problem of facility location, given a set of points $P \subseteq \cX$, called \emph{demands}, a distance function $\dist(\cdot,\cdot)$ and fixed cost $f>0$, conventionally called the \text{facility cost}, asks to find a set of points $\cC$ that minimizes the objective 
$$\minu{\cC \subseteq \cX}{\sum_{p \in P} \minu{c \in \cC}{\ \dist(p,c)} + \abs{\cC}f}.$$

Informally, it seeks a set of points such that the \emph{cumulative} cost of serving the demands (known as \emph{service cost}), which is $\dist(p,c)$  and opening new facilities $f$, is minimized. Online Facility Location is the variant of the above problem in the streaming setting, wherein the facility assignments and service costs incurred are irrevocable. 
That is to say, once a point is assigned to a facility, it cannot be reassigned to a different facility at a later point in time, even if the newer facility is closer. 
A simple and popular algorithm to this problem is by Meyerson \cite{meyerson2001online}, wherein upon receiving a point, it calculates its distance to the nearest facility and flips a coin with bias equal to the distance divided by facility cost. 
If the outcome is heads (or $1$), it opens a new facility, otherwise the nearest point serves this demand and it incurs a service cost, equal to the distance. 
From here on, we abuse notation and use \ofl to refer to the algorithm of Meyerson \cite{meyerson2001online}.

\subsection{Algorithm}
Our algorithm for exponential decaying streams is a variant of the popular $k$-median clustering algorithm \cite{braverman2011streaming,charikar2003better}, which uses \ofl \ as a sub-routine. We first briefly discuss the algorithm of  \cite{braverman2011streaming} and then elucidate on how we adapt this to exponential decay streams. The algorithm operates in \emph{phases}, where in each phase it maintains a \emph{guess}, denoted by $L$, to the lower bound on optimal cost. It then uses this guess to instantiate the \ofl \ algorithm of \cite{meyerson2001online} on a set of points in the stream. If the service cost of \ofl \ grows high or the number of facilities grows large, it infers that the guess is too low and triggers a \emph{phase change}. It then increases the guess by a factor of $\beta$ (to be set appropriately) and the facilities are put back at the start of the stream and another round of \ofl \ is run. 

\paragraph{Notation.} We first define and explain some key quantities. The \emph{aspect ratio} of a set is defined as the ratio between the largest distance and the smallest non-zero distance between any two points in the set. We use $\Delta$ to denote the aspect ratio of the stream $P$. For simplicity of presentation, we assume that the minimum non-zero distance between two points is at least $1$. We define $W$ as the total weight of the first $\np$ points in the stream divided by the minimum weight. Suppose the stream starts at $t=z$, then for any $h=\Omega(1)$,
 
$$W = \frac{1}{2^{z/h}}\sum_{t=z}^{h \log (\Delta+1)} 2^{t/h}  =  \frac{\Delta}{2^{1/h} - 1} = \Theta(h \Delta).$$

For a set $P \subseteq (\cX,\dist)$, we use $\opt_k(P)$ to denote the optimal $k$-median clustering cost for the set. For two sets $P$ and $S$, we use $\cost(P,S)$ to denote the cost of clustering $P$ with $S$ as medians. Whenever we use $\opt$, it corresponds to the optimal cost of $k$-median clustering of the stream seen till the point in context. We use \ram \ to denote an offline constant $\cram$-approximate $k$-median clustering algorithm in the random access model (RAM). Given a set of points $P$ and a positive integer $k$, \ram \ outputs $(\cC,\lambda)$, where $\cC$ is a set of $k$ points and $\lambda = \cost(P,\cC) \leq \cram  \cdot \opt_k(P)$.
\paragraph{Our Algorithm.} Our algorithm, inspired from  \cite{charikar2003better,braverman2011streaming}, works in phases. We however have important differences. Each of our phases are again sub-divided into two \emph{sub-phases}. In the first sub-phase we execute \ofl \ same as \cite{charikar2003better,braverman2011streaming} and after each point we check if the cost or the number of facilities is too large. If this is indeed the case, we trigger a phase change. However, if we read $\np$ points in a phase, then we move on to the second sub-phase of the algorithm. Here we simply count points and store them verbatim. Upon reading $k + \npv$ points, we trigger a phase change.  The intuition for this sub-phase is that a phase change is triggered when \opt \ increases by a factor of $\beta$.  After $\np$ points, subsequent points are so heavy relative to points of the previous phase that any service cost will be large enough to ensure \opt \ has increased. Therefore, we restrict the algorithm to read at most $\np + k+ \npv$ points in a single phase.  When we start a new phase, we cluster the existing facility set to extract exactly $k$ points using an off-the-shelf constant approximate \ram \ algorithm and continue processing the stream. We present the above idea in full in Algorithm~\ref{alg:expdecay}. We now state our main theorem for exponential decay streams.

\begin{theorem}
There exists a streaming algorithm that given a stream $P$ of exponential decaying points with aspect ratio $\Delta$ and half-life $h$, produces an $\bigO{1}$-approximate solution to k-median clustering. The algorithm runs in $\bigO{nk\log (h \Delta)}$ time and uses $\bigO{k \log (h\Delta) + h}$ space.
\end{theorem}

\begin{algorithm}[!htb]
\caption{k-median clustering in exponential decay streams}
\label{alg:expdecay}
\begin{algorithmic}[1]
	\REQUIRE $k$, stream $P$, an offline constant approximate $k$-median clustering algorithm \ram.
	\STATE $L \leftarrow 1$, $\cC \leftarrow \emptyset$ 
	\WHILE{solution not found}
       \STATE $i \leftarrow 0$, $\cost \leftarrow 0$, $f \leftarrow \dfrac{L}{k(1+\np)}$.
        \WHILE{stream not ended} 
            \STATE $p \leftarrow $ next point from stream \\
            \STATE $q \leftarrow$ closest point  to $p$ in $\cC$  \\
            \STATE $\sigma\gets\br{\min \br{\dfrac{w(p) \cdot \dist(p,q)}{f},1}}$
            \IF[do with probability $\sigma$]{probability $\sigma$}
                \STATE $\cC \leftarrow \cC \cup \bc{p}$
            \ELSE
                \STATE $\cost \leftarrow \cost + w(p) \cdot \dist(p,q)$
                \STATE $w(q) \leftarrow w(q) + w(p)$
                \ENDIF
            \STATE $i \leftarrow i+1$
            \IF[cost or number of facilities too large]{$\cost > \gamma L \text{ or } \abs{\cC} > (\gamma - 1)k(1+\log W)$}
                \STATE break and raise flag
                \COMMENT{trigger phase change}
            \ELSIF[second sub-phase]{$i \geq \np$}  
                \FOR[count points and store them verbatim]{$l = 1$ to $\npv+k$}
                    \STATE $p \leftarrow $ next point from stream \\
                    \STATE $\cC \leftarrow \cC \cup \{p\}$
                \ENDFOR
                \STATE break and raise flag
            \ENDIF
            	 \ENDWHILE
            \IF[phase change]{flag raised} 
            \STATE $(\cC,\lambda) \leftarrow \text{\ram \ }(\cC,k)$
            \COMMENT{cluster existing facilities}
             \STATE $L \leftarrow \max \br{\beta L,\frac{\lambda}{\cram \gamma}}$
                \ELSE
                \STATE Declare solution found
          \ENDIF
          \STATE $(\cC,\lambda) \leftarrow \text{\ram \ }(\cC,k)$
         \ENDWHILE 
	\ENSURE $ \cC, \cost$
	\end{algorithmic}
\end{algorithm}

\subsection{Analysis}
We first analyze the service cost and space complexity of \ofl. 
For the $t^\text{th}$ point in the stream $p_t$, the weight of $p_t$, denoted $w(p_t)$, is  $w(p_t)= 2^{t/h}$.
The following two lemmas will establish bounds on the service cost and number of facilities of \ofl.
\begin{lemma} 
\label{lem:oflcost}When \ofl \ is run on a stream of $n$ points with exponentially decaying weights, with facility cost $f = \frac{L}{k(1+\log W)}$ where $L > 0$, it produces a service cost of at most $6\opt_k(P) + 2L$ with probability at least $1/2$.
\end{lemma}
\begin{proof}
The proof follows the standard analysis of Online Facility Location. Let $P$ is the set of points read in a phase. Instead of looking at $|P|$ distinct points with varying weights, we view it as \emph{repeated} points of unit or minimum weight. The total number of points is therefore at most $W = \Theta(h\Delta)$. 

We remind the reader that $\opt_k(P) = \minu{K \subseteq P, |K| = k}{ \sum_{p \in P} \minu{y \in K}{ \ \dist(p,y)}}$ is the optimal cost and $\cost(P)$ is the total service cost incurred by \ofl. Let $\cC^*$ be the set of corresponding facilities allocated by $\opt$, and  $c^*_i$'s denote the optimum $k$ facilities where $i \in [k]$ and $C_i^*$ the set of points from $P$ served by the facility $c^*_i$.  Let $A_i = \sum_{x \in C_i^*} \dist(x,c_i^*)$ be the service cost of $C_i^*$. We now further partition each region into \emph{rings}. Let $S_i^1$ be the first ring around $c_i^*$ that contains half the nearest points in $C_i^*$. Formally, $S_i^1 = \argminu{K, \abs{K} = \abs{C^*_i}/2}{\sum_{x \in K} \dist(x,c_i^*)}$. Furthermore, $S_i^2$ is the second ring around $c^*_i$ containing one-quarter of the points in $C_i^*$ and so on. Therefore, we can inductively define $S_i^j = \argminu{K, \abs{K} = \abs{C^*_i}/2^j}{\sum_{x \in K \backslash \cup_{l=1}^{j-1}S_i^l} \dist(x,c_i^*)}$. 
Note that $S_i^j$ may be not be uniquely identifiable, but their existence suffices for the sake of analysis. Let $A_i^j = \sum_{x \in S_i^j} \dist(x,c_i)$ be the cost of set $S_i^j$.  For a point $p$,  use $\dist_p^*$ and $\dist_p$ for its optimal cost and cost incurred in the algorithm respectively.

We look at two cases. In the first case, suppose each region has a facility open; let the facility of $S_i^j$ be $s_i^j$. We look at the cost incurred by subsequent points arriving in this region. Consider the set $S_i^j$ and let $q$ be a facility in $S_i^j$. A subsequent point $p$ incurs a cost $\dist_p = \dist(p,q)$. By triangle inequality, we have $\dist_p \leq \dist_p^* + \dist_q^*$. By definition of $S_i^j$, we have $\dist_q^* \leq \dist_z^*$ for any point $z \in S_i^{j+1}$.
We sum over all $z$ in $S_i^{j+1}$ and get $\dist_q^* \leq \frac{A_i^{j+1}}{\abs{S_i^{j+1}}}$. We therefore get $\dist_p \leq \dist^*_p  + \frac{A_i^{j+1}}{\abs{S_i^{j+1}}}$. Summing over all points is $S_i^j$, we get $\cost(S_i^j, s_i^j) \leq A_i^j + \frac{\abs{S_i^j} A_i^{j+1}}{{\abs{S_i^{j+1}}}} = A_i^j + 2 \cdot A_i^{j+1}$. Summing over all $j$'s, we get $\cost(C^*_i, c_i^*) \leq 3 A_i$. Finally, summing over $i$'s, we get that in the first case $ \cost(P,\cC^*) \leq 3\opt_k(P)$. We now look at the second case wherein each region has a facility open. The number of points is at most $W$, therefore, the number of regions is at most $k(1 + \log(W))$. The expected service cost incurred by a region before opening a facility is  at most $f$ (See Fact $1$, \cite{lang2017online}). Therefore, the total service cost $\leq f \, k(1 + \log(W)) = L$. Combining the two cases, we get that $\cost(P,\cC^*) \leq 3\opt_k(P) + L$. Note that when we store points verbatim, we do not incur any service cost. With a simple application of Markov inequality, we get that with probability at least $1/2$, $ \cost(P,\cC^*)\leq 6\opt_k(P) + 2L$.
\end{proof}

\begin{lemma}
\label{lem:oflspace}
When \ofl \ is run on a stream of $n$ points with exponentially decaying weights, with facility cost $f  = \frac{L}{k(1+\log W)}$ where $L>0$, the number of facilities produced is at most $(2 + \frac{6}{L}\opt_k(P))k(1+\log W)$, with probability at least $1/2$.
\end{lemma}
\begin{proof}
Considering the points as repeated points of minimum weight, the total number of points is at most $W$  and the total number of regions is at most $k(1 + \log W)$.  One facility in each region gives us $k(1 + \log  W)$ facilities. After opening a facility in a region, each subsequent point has probability $\frac{\dist_p}{f}$ to open a facility. Therefore, the expected number of facilities is $\sum_p \frac{\dist_p}{f}$. We showed in Lemma \ref{lem:oflcost} that $\sum_{p} \dist_p \leq 3 \ \opt_k(P)$. Hence, the expected number of facilities is at most $\frac{3  \opt_k(P)}{f} = \frac{3 \opt_k(P) k (1 + \log W)}{L}$. A simple application of Markov's inequality completes the proof. 
\end{proof}

\paragraph{$k$-median clustering.} We now state some key lemmas that will help us establish that the algorithm produces a $\bigO{1}$ approximation to the $k$-median clustering cost. We then show how these come together and present the detailed guarantees in Theorem \ref{thm:k-median-exp-decay-detail}.

\begin{lemma}
\label{lem:2}
At every phase change, with probability at least $1/2$, $\opt_k(P) > L$ if $\beta \leq 2$ and $\gamma \geq 9$.
\end{lemma}
\begin{proof}
The phase change is triggered in two ways, either the cost or the number of facilities grows large (more precisely, cost more that $\gamma L$ or the number of facilities greater than $(\gamma-1)k(1+\log W)$), or we read too many points. Let us look at the first case. Assume that $L \geq \opt_k(P)$, then from Lemma \ref{lem:oflcost} and \ref{lem:oflspace}, we get that with probability at least $1/2$, $\cost \leq 8L$ and the number of facilities is $\leq 8k(1+\log W)$ respectively. However with $\gamma \geq 9$, neither of the two conditions are met and therefore the premise that a phase change was triggered gives us a contradiction. Hence, in the first case, we get $L < \opt_k(P)$ with probability at least $1/2$.

In the other case, we store points exactly (incurring no additional cost). The only danger in this case is performing a phase change too early (before \opt \  has doubled). Let $ \underbar \opt$ be the value of \opt \ at the beginning of the phase, which we assume starts at time $t = z$. Since points cannot be at distance greater than $\Delta$, then
\begin{align*}
    \underbar \opt & \leq \Delta (1 + 2^{1/h} + \ldots + 2^{z/h}) \\
    & \leq \Delta \frac{2^{(z+1)/h}-1}{2^{1/h}-1}.
\end{align*}

Now let $\overline\opt$ be the value of \opt  \ after terminating the phase (which occurs after reading $k+\npv$ distinct points after the initial $\np$ points of the phase). We must prove that $\overline \opt \geq 2\underbar \opt$. Observe that after reading $k + \npv$ distinct points, we must cluster at least $h$ points across a distance of at least $1$ (since we can have at most $k$ centers). The weights of these points begin at $2^{(z+\np+1)/h}$. Therefore,
\begin{align*}
     \overline \opt &\geq \underbar\opt +  \sum_{i=z+\np}^{z+\np+\npv}2^{i/h}\\
    &= \underbar\opt +  \frac{2^{(z+h\log_2(\Delta) + h)/h} - 2^{(z+h\log_2(\Delta))/h}}{2^{1/h} - 1} \\
    & \geq \underbar \opt + \Delta\left(\frac{2^{(z+1)/h}-1}{2^{1/h}-1}\right)\\
    & \geq 2 \underbar  \opt,
\end{align*}
where the second inequality follows from straightforward arithmetic. 
Let $L'$ be the value of $L$ in the previous phase. 
Thus,
\begin{align*}
     \overline\opt \geq 2 \underbar  \opt > 2 L' = \frac{2}{\beta}L,
\end{align*}
where the second inequality holds with probability at least $1/2$, as justified above. 
Setting $\beta \leq 2$ completes the proof.
\end{proof}

\begin{lemma}
\label{lem:1}
At any part in the algorithm, we have $\cost(P,\cC) \leq \br{\gamma + \frac{1 + \cram \beta}{\beta-1}}L$.
\end{lemma}
\begin{proof}
We know that the increase of $\cost(P,\cC)$ in the current phase is upper bounded by the variable $\cost$ (see Algorithm \ref{alg:expdecay}). 
In a single phase, we have $\cost \leq \gamma L$. 
Therefore, outside the phase loop, we just need to show that it is at most $\frac{1+\cram \beta}{\beta-1}L$. 
Note that it changes only by the $\ram$ \ algorithm, which incurs cost of $\lambda \leq \cram \gamma L$. 
Suppose that it holds in the previous phase and let $L'$ be the value of $L$ in the previous phase. 
Then the cost outside the loop is $\gamma  L' + \frac{1+\cram \beta}{\beta-1}L' + \lambda  \leq \frac{1+\cram \beta}{\beta-1}L$, which finishes the proof.
\end{proof}
\begin{lemma}
\label{lem:3}
With probability at least $1/2$, $L \leq \br{1 + \frac{1}{\gamma} + \frac{1 + \cram \beta}{\gamma(\beta-1)}}\opt_k(P)$.
\end{lemma}

Let $L'$ and $\cC'$ denote the values of $L$ and $\cC$ in the previous phase. We condition on the event that  $L' < \opt_k(P)$, which we know from Lemma \ref{lem:2} occurs with probability at least $1/2$. From the update equation of $L$, we either have $L = \beta L'$ or $L = \frac{\lambda}{\cram \gamma}$. In the first case, we directly get $L \leq \beta \opt_k(P)$. With $\beta \leq 2$, we get the claim of the lemma.
We now look at the second case, where we have $\gamma \cram L  \leq \lambda \leq \cram \opt_k(\cC')$ from the guarantee of the \ram \ algorithm. It is easy to see that $\opt_k(\cC') \leq \opt_k(P) + \cost(P,\cC')$ by a simple application of triangle inequality on all the points. Moreover, from Lemma \ref{lem:1}, we have $\cost(P,\cC') \leq \br{\gamma + \frac{1 + \cram \beta}{\beta-1}}L' \leq \br{\gamma + \frac{1 + \cram \beta}{\beta-1}} \opt_k(P)$. Combining these, we get $L \leq \br{\frac{1}{\gamma} + 1 + \frac{1 + \cram \beta}{\gamma(\beta-1)}}\opt_k(P)$.

We now restate the theorem for the exponential decay model but tailored to Algorithm \ref{alg:expdecay} with all the algorithmic details precisely stated.
\begin{theorem}
\label{thm:k-median-exp-decay-detail}
Let $P$ be a stream of $n$ points with exponential decaying weights parametrized by the half-life parameter $h$ and let $k$ be some positive integer. Algorithm \ref{alg:expdecay} run with $\beta \leq 2, \gamma \geq 9 ,W = \bigO{h \Delta}$ on the stream $P$ outputs $k$ points, which produce an $\bigO{1}$ approximation to the optimal cost of $k$-median clustering on $P$ with high probability. The algorithm runs in time $\bigO{nk\log W}$ and uses space $\bigO{k \log W + h}$.
\end{theorem}
\begin{proof}
Combining Lemma \ref{lem:1} and \ref{lem:3}, we get that 
$$\cost(P,\cC) \leq \br{\gamma + \frac{1 + \cram \beta}{\beta-1}}\br{\frac{1}{\gamma} + 1 + \frac{1 + \cram \beta}{\gamma(\beta-1)}}\opt_k(P).$$
Setting $\beta = 2, \gamma = 10$ and $\cram = 3$ gives us that $\cost(P,\cC) \leq 40 \opt_k(P)$.
 
We emphasize that we give a \emph{streaming} guarantee, that is, given a fixed point in the stream, it will hold for all the points seen till then. Note that in the proofs of Lemma \ref{lem:2} and \ref{lem:3}, we only need that the random event hold with probability at least $1/2$ \emph{only} in the previous phase. We can therefore amplify the probability of success by running $\log(1/\delta)$ parallel instances to get the bounds to hold with probability at least $1-\delta$. The space bound of the algorithm is $\bigO{k \log W + h} = \bigO{k \log (h \Delta) + h}$, which simply follows from the condition in the algorithm that we don't allow the number of facilities to grow beyond $\bigO{k(1+\log(W)}$ combined with the fact that we store $k+\npv$ points verbatim in the second sub-phase.
\end{proof}

\paragraph{Extensions.} As in \cite{lang2017online}, our algorithm can easily be extended to other distance functions that satisfy the approximate triangle inequality (see Definition \ref{def:ati}). In particular, we get constant approximate algorithms for  $k$-means clustering and $M$-estimators in the exponential decay model.

\bibliography{refs}
\bibliographystyle{alpha}

\end{document}